\documentclass[10pt,conference]{IEEEtran}

\usepackage[utf8]{inputenc}
\usepackage{amsmath,amsfonts,amssymb}
\usepackage{graphicx}
\usepackage{float}
\usepackage{booktabs}
\usepackage{multirow}
\usepackage{array}
\usepackage{xcolor}
\usepackage{subcaption}
\usepackage{amsthm} 
\usepackage{algorithm}
\usepackage{mathtools}
\usepackage[noend]{algpseudocode}  

\usepackage{cite}
\usepackage{hyperref}

\usepackage{braket}
\usepackage{mathrsfs}

\newtheorem{lemma}{Lemma}

\title{Efficient and Practical Black-Box Verification of Quantum Metric Learning Algorithms}

\author{

    \IEEEauthorblockN{Ahmed Shokry}
    \IEEEauthorblockA{
        Computer Science and Engineering\\
        The Pennsylvania State University\\
        University Park, USA\\
        ams12545@psu.edu
    }
    \and
    \IEEEauthorblockN{Movahhed Sadeghi}
    \IEEEauthorblockA{
        Computer Science and Engineering\\
        The Pennsylvania State University\\
        University Park, USA\\
        mus883@psu.edu
    }
    \and
    \IEEEauthorblockN{Mahmut Kandemir}
    \IEEEauthorblockA{
        Computer Science and Engineering\\
        The Pennsylvania State University\\
        University Park, USA\\
        mtk2@psu.edu
    }

}

\begin{document}

\maketitle

\begin{abstract}
\boldmath
Quantum metric learning enhances machine learning by mapping classical data to a quantum Hilbert space with maximal separation between classes. This separation is important for the performance of downstream quantum tasks. However, on current Noisy Intermediate-Scale Quantum (NISQ) hardware, this mapping (embedding) process itself is prone to errors and, in untrusted settings, could be fundamentally incorrect. Therefore, verifying that a quantum embedding model successfully achieves its promised separation is essential to ensure the correctness and reliability of the entire quantum machine learning pipeline.

In this paper, we propose a practical black-box verification protocol to audit the performance of quantum metric learning models. We define a setting with two parties: a powerful but untrusted \textit{prover}, who claims to have a parameterized unitary circuit that embeds classical data from different groups with a guaranteed angular separation, and a limited \textit{verifier}, whose quantum capabilities are restricted to performing only basic measurements. The verifier has no knowledge of the implementation of the prover, including the structure of the model, its parameters, or the details of the prover measurement setup. To verify the separation between different data groups, the proposed algorithm must overcome two key challenges. First, the verifier is ignorant of the prover's implementation details, such as the optimization cost function and measurement setup. Consequently, the verifier lacks any prior information about the expected quantum embedding states for each group. Second, the destructive nature of quantum measurements prevents direct estimation of the separation angles. Our algorithm successfully overcomes these challenges, enabling the verifier to accurately estimate the true separation angles between the different groups.

We implemented the proposed protocol and deployed it to verify the QAOAEmbedding models. The results from both theoretical analysis as well as practical implementation show that our proposal effectively assesses embedding quality and remains robust in adversarial settings. Our work thus provides a critical tool for validating the correctness and quality of quantum embeddings.
\end{abstract}


\section{Introduction}
Quantum Machine Learning (QML) has emerged as a promising application for near-term quantum computers, offering potential advantages in tasks such as classification and generative modeling~\cite{wittek2014quantum, biamonte2017quantum, schuld2018supervised}. Among various QML approaches, those based on variational quantum circuits are widely studied due to their adaptability to noisy intermediate-scale quantum (NISQ) hardware \cite{Havlicek2019, Benedetti2019b}. A typical variational QML classifier comprises two key components: a quantum feature map that encodes classical data into a high-dimensional Hilbert space, and a variational ansatz that processes these states for classification~\cite{havlivcek2019supervised}.

Traditional QML often focuses on optimizing the measurement and classification circuit while treating data embedding as a fixed, non-trainable component~\cite{Schuld2019}. In contrast, a more recent and powerful paradigm, known as \textit{Quantum Distance Metric Learning}, optimizes the quantum feature map itself to maximize the separation between data classes within the Hilbert space~\cite{lloyd2020quantum,benedetti2019parameterized,schuld2020circuit,farhi2018classification}. This approach is grounded in the principle that creating well-separated quantum representations of data allows the use of simpler and more robust classification schemes. The choice of distance metric directly informs the optimal classifier: for instance, maximal separation under the trace distance enables the application of the optimal Helstrom measurement for distinguishing quantum states~\cite{Helstrom1976}, while separation defined by the Hilbert-Schmidt distance naturally leads to classifiers based on quantum fidelity~\cite{Blank2019}.

The importance of quantum metric learning extends beyond theoretical interest. By shifting the computational burden from measurement optimization to embedding optimization, this approach reduces the quantum circuit depth required for effective classification. This reduction is a critical advantage for NISQ-era devices, which are constrained by limited coherence times~\cite{Killoran2019}. Furthermore, the ability to learn high-quality embeddings is a key pathway towards potential quantum advantages in machine learning (ML)~\cite{Harrow2017}. This potential is particularly relevant for kernel methods, where quantum computers can efficiently evaluate complex, high-dimensional kernels utilizing the natural geometry of Hilbert space, without the need to explicitly construct the feature map~\cite{lloyd2014quantum,lloyd2016quantum,rebentrost2014quantum}.

However, the quality of the learned embedding must be verifiable to ensure reliable performance. Although recent work has demonstrated the feasibility of training quantum embeddings using hybrid quantum-classical optimization \cite{Bergholm2018}, the question of how to \textit{verify} that the claimed quantum advantage in metric learning remains open. This verification is particularly important given the probabilistic nature of quantum measurements and the presence of noise in current quantum hardware.

In this paper, we introduce a practical black-box verification algorithm designed to audit the performance of quantum metric learning models. We formalize this within a two-party framework: a computationally powerful but untrusted prover $\mathcal{P}$ , who claims to have a parameterized quantum circuit that can generate embeddings with a guaranteed angular separation between data classes, and a limited verifier $\mathcal{V}$. The quantum capabilities of the verifier are restricted to performing basic measurements and do not have internal knowledge of the implementation of the prover, such as the model architecture, its trained parameters, or the specifics of the measurement setup used. To this end, verification must address two fundamental challenges inherent in this setting. First, the verifier ignorance of the prover's model details prevents any prior knowledge of the expected quantum embedding states. Second, the destructive nature of quantum measurements obstructs the direct estimation of the separation angles. Our proposed algorithm overcomes these challenges, enabling the verifier to accurately estimate the true separation angles, providing a robust technique to validate the prover claims without knowing the details of the internal circuit model.

We implemented and deployed the proposed verification algorithm to evaluate the PennyLane QAOAEmbedding variational quantum model~\cite{bergholm2018pennylane, farhi2014quantum}. 
The results from our theoretical analysis as well as practical implementation show that the proposed protocol can effectively evaluate quantum metric learning models and performs well in adversarial scenarios, demonstrating its utility as a verification algorithm for quantum embedding quality.

The remainder of this paper is organized as follows. Section~\ref{background} reviews the fundamentals of quantum metric learning and statistical interactive proofs. Section~\ref{methodology} presents our protocol for verifiable quantum metric learning and embeddings. 
Section~\ref{experiment} gives our experimental results. Section~\ref{discussion} discusses the proposed algorithm and how to extend it, and finally Section~\ref{conclusion} concludes the paper.


\section{Background}
\label{background}
\subsection{Quantum Metric Learning}
This section gives a background on quantum metric learning for QML by optimizing the quantum feature map that embeds classical data into a high-dimensional Hilbert space. The key idea is to maximize the separation between the data from the different sets in the Hilbert space, which simplifies the subsequent machine learning tasks. 

The process begins by uniformly sampling data points from two classes, denoted as class $A$ and class $B$. For class $A$, we sample $M_a$ data points $\{a_1, \ldots, a_{M_a}\}$, and for class $B$, we sample $M_b$ data points $\{b_1, \ldots, b_{M_b}\}$. These data points are classical inputs that need to be embedded into quantum states.

Each classical data point $x$ (where $x$ can be $a_i$ or $b_j$) is mapped to a quantum state $|x\rangle$ using a parametrized quantum circuit $U(x, \theta)$. This circuit is a unitary transformation that depends on both the input data $x$ and a set of trainable parameters $\theta$. The embedding is performed by applying $U(x, \theta)$ to an initial state $|0\ldots 0\rangle$, resulting in the quantum state $|x\rangle = U(x, \theta) |0\ldots 0\rangle$. 

The embedded states from each class are used to construct two density matrices (ensembles) representing the classes in the Hilbert space. For class $A$, the ensemble is given by:
    \[
    \rho = \frac{1}{M_a} \sum_{a \in A} |a\rangle \langle a|,
    \]
    and for class $B$, the ensemble is:
    \[
    \sigma = \frac{1}{M_b} \sum_{b \in B} |b\rangle \langle b|.
    \]
These density matrices capture the statistical properties of the embedded data points for each class.

The separation between the two ensembles $\rho$ and $\sigma$ is quantified using a distance metric in Hilbert space. The trace distance ($\ell_1$)
\[D_{\text{tr}}(\rho, \sigma) = \frac{1}{2} \text{tr}(|\rho - \sigma|)\]
is the quantum analog of the $\ell_1$ distance. Maximizing this distance ensures optimal separation for the Helstrom classifier~\cite{helstrom1969quantum}.
The Hilbert-Schmidt Distance ($\ell_2$)~\cite{coles2019strong}
\[D_{\text{hs}}(\rho, \sigma) = \sqrt{\text{tr}((\rho - \sigma)^2)}\]
corresponds to the $\ell_2$ distance. Maximizing this distance is well-suited for use in the fidelity classifier. 

The distance between $\rho$ and $\sigma$ can be estimated using quantum circuits. For the Hilbert-Schmidt distance, the SWAP test or inversion test circuits are employed. The SWAP test measures the overlap between two quantum states $|\psi\rangle$ and $|\phi\rangle$ by estimating $|\langle \psi | \phi \rangle|^2$. By applying the SWAP test to pairs of states from $\rho$ and $\sigma$, the terms $\text{tr}(\rho \sigma)$, $\text{tr}(\rho^2)$, and $\text{tr}(\sigma^2)$ can be computed, which are needed to evaluate $D_{\text{hs}}(\rho, \sigma)$.
\begin{equation}
    D_{\mathrm{hs}}(\rho,\sigma)=\sqrt{\operatorname{tr}(\rho^{2})+\operatorname{tr}(\sigma^{2})-2\,\operatorname{tr}(\rho\sigma)}.
\end{equation}
The parameters $\theta$ of the embedding circuit $U(x, \theta)$ are updated to maximize the chosen distance metric. This is done using classical optimization techniques such as gradient descent~\cite{ruder2016overview}. The cost function $C$ to be minimized is often a function of the distance between $\rho$ and $\sigma$, 
where minimizing $C$ corresponds to maximizing the distance. The optimization process iteratively adjusts $\theta$ to reduce overlap between the ensembles, thus increasing their separation.

The procedure repeats until the average angle between the states of the two classes is as large as possible, ideally $\pi/2$, while the angle between the states of the same class is as small as possible, ideally 0. At this point, the ensembles $\rho$ and $\sigma$ are maximally separated in Hilbert space, and the optimal measurement for classification (e.g., Helstrom or fidelity measurement) can be applied with minimal error. Figure~\ref{fig:qml} demonstrates the separation of data class embeddings achieved through quantum metric learning.

\begin{figure}
    \centering
    \includegraphics[width=0.8\linewidth]{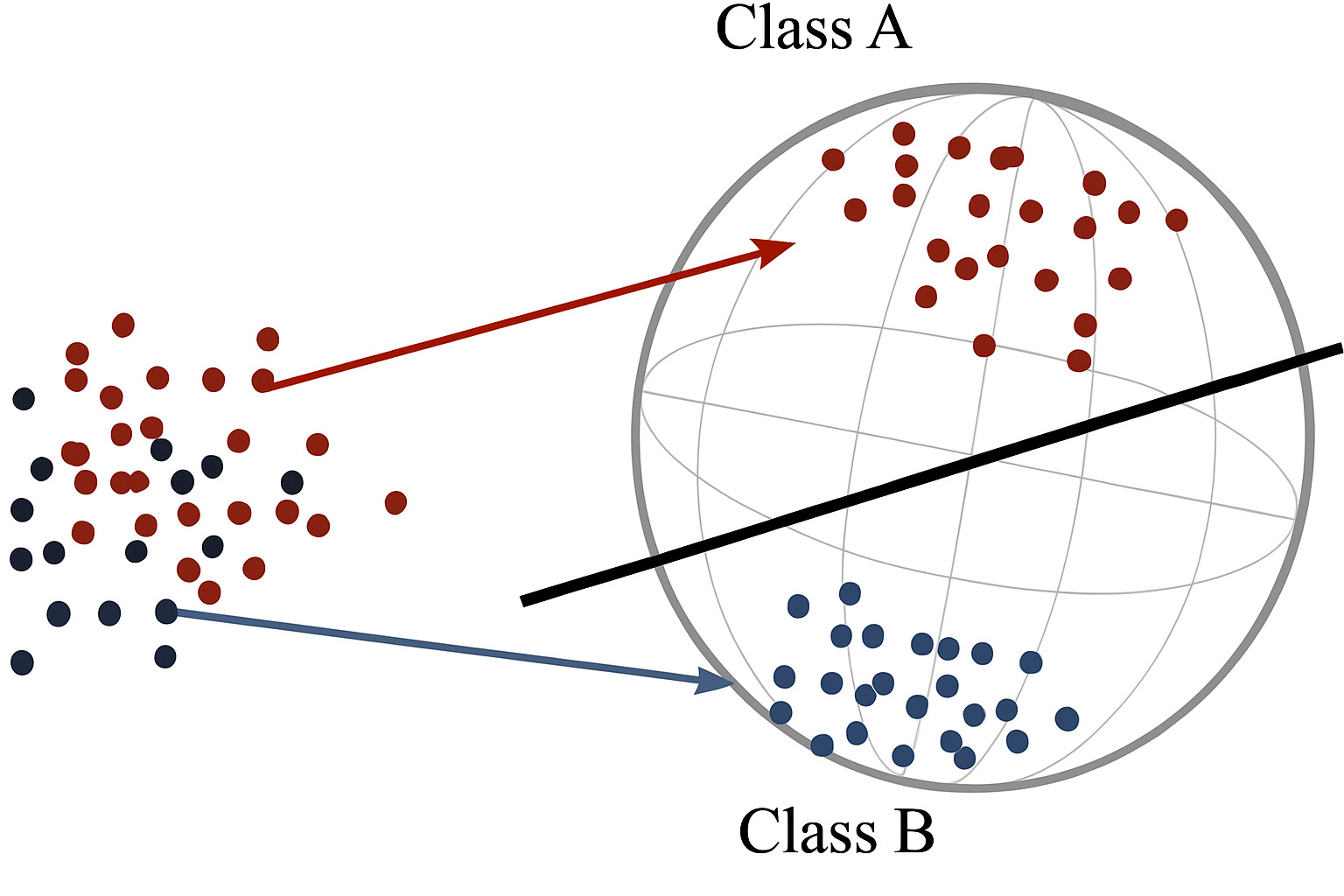}
    \caption{Example of quantum metric learning. \label{fig:qml}}
    \vspace{-0.1in}
\end{figure}


\subsection{Statistical Interactive Proofs}
A statistical interactive proof involves a computationally unbounded \emph{prover} claiming a statement is true, and a computationally bounded \emph{verifier} who interacts with the prover to verify this claim. The protocol must satisfy:
\textbf{(1) Completeness}: If the claim is true, an honest prover convinces the verifier to accept with high probability. \textbf{(2) Soundness}: If the claim is false, no cheating prover can make the verifier accept except with negligible probability. These properties ensure that honest proofs are accepted while dishonest proofs are rejected with high confidence~\cite{goldwasser1989knowledge}.

\section{Verification Procedure}
\label{methodology}
We begin by formally describing the proposed verification protocol and outlining the key challenges. We then present the protocol’s detailed construction, followed by an analysis of its completeness and soundness.
\subsection{Formal Description and Challenges}
Consider an oracle $\mathcal{O}$ that provides an infinite stream of data samples. Assume, without loss of generality, that each sample is a two-dimensional vector $\vec{x}= (x_1, x_2)$ with a binary label $y \in \{0,1\}$, such that the complete data specification is $(\vec{x}, y)$. A prover $\mathcal{P}$ claims to possess a parameterized unitary circuit $U(\vec{x},\theta)$ that acts as an optimal quantum embedding. Specifically, when applied to data from this oracle, $U(\vec{x},\theta)$ is claimed to produce quantum state embeddings that separate the two classes with maximal angular separation in Hilbert space. The verification framework involves two distinct parties with asymmetric knowledge and capabilities:
\begin{itemize}
    \item \textbf{Verifier ($\mathcal{V}$)}: Has exclusive access to the data oracle $\mathcal{O}$ but possesses \emph{no} knowledge of the prover's unitary $U(\vec{x},\theta)$ including its architecture, the specific parameter values $\theta$, or the measurement device used. The verifier's quantum capabilities are restricted to performing only basic quantum measurements.
    \item \textbf{Prover ($\mathcal{P}$)}: Possesses the full implementation of $U(\vec{x},\theta)$ but has \emph{no} direct access to the verifier's data oracle $\mathcal{O}$.
\end{itemize}

The central task for the verifier is to determine, using only its limited quantum resources and data access, whether the prover's embedding $U(\vec{x},\theta)$ indeed maintains the promised maximal angular separation between quantum states representing different classes, despite having zero internal knowledge of the prover's implementation.

We formalize the problem as a verification system $\mathcal{Q}=(\mathcal{P},\mathcal{V})$. Consider a dataset $\mathcal{D}$ composed of two disjoint classes, $\mathcal{D}_0$ and $\mathcal{D}_1$, such that $\mathcal{D} = \mathcal{D}_0 \cup \mathcal{D}_1$. Let $\Psi$ and $\Phi$ be index sets for the data samples in $\mathcal{D}_0$ and $\mathcal{D}_1$, respectively. Thus, the classical data vectors are $\{\vec{x}_i \in \mathcal{D}_0\}_{i \in \Psi}$ with the label $y_i=0$ and $\{\vec{x}_j \in \mathcal{D}_1\}_{j \in \Phi}$ with the label $y_j=1$. The prover $\mathcal{P}$ possesses a parameterized unitary operation $U(\vec{x}, \theta)$ that encodes a classical input vector $\vec{x}$ into a quantum state. The prover uses this circuit to prepare the following quantum embeddings:
\begin{itemize}
    \item For class 0: $\{\ket{\psi_i} = U(\vec{x}_i, \theta)\ket{0} \mid i \in \Phi\}$
    \item For class 1: $\{\ket{\phi_j} = U(\vec{x}_j, \theta)\ket{0} \mid j \in \Psi\}$
\end{itemize}


 
The prover claims that their embedding model achieves the optimal separation defined in the Lemma~\ref{lem:orth}. This optimality is concretely promised as orthogonality between classes: any state from class 0 is orthogonal to any state from class 1 ($\langle\psi_i|\phi_j\rangle=0$), ensuring a maximal angular separation of $\pi/2$ in the Hilbert space.

\begin{lemma}[Optimal Quantum Metric Learning]\label{lem:orth}
Given two sets of quantum states $A = \{\ket{a_i}\}_{i=1}^{M_a}$ and $B = \{\ket{b_j}\}_{j=1}^{M_b}$ prepared by a parameterized quantum circuit $U(\vec{x}, \theta)\ket{0}$, under the optimization procedure that:
\begin{enumerate}
    \item Maximizes the average fidelity between intra-set states: $\max_{\theta} \frac{1}{|A|^2} \sum_{i,j} |\langle a_i|a_j\rangle|^2$ and similarly for set $B$.
    \item Minimizes the average fidelity between inter-set states: $\min_{\theta} \frac{1}{|A||B|} \sum_{i,j} |\langle a_i|b_j\rangle|^2$.
\end{enumerate}
The optimal solution achieves the following:
\begin{enumerate}
    \item Intra-set angles $\Theta_{AA}, \Theta_{BB} \leq \epsilon$ (for small $\epsilon$, indicating that states within a class are nearly identical).
    \item Inter-set angles $\Theta_{AB} \rightarrow \pi/2$ (indicating that states from different classes are maximally separated).
\end{enumerate}
\end{lemma}

\begin{proof}
The proof follows from the variational principles of quantum state discrimination \cite{helstrom1969quantum}. In intra-set optimization, the objective function $\max_{\theta} \frac{1}{|A|^2} \sum_{i,j} |\langle a_i|a_j\rangle|^2$ is maximized when all inner products $|\langle a_i|a_j\rangle|^2$ are 1, meaning that all states $\ket{a_i}$ are identical up to a global phase. Geometrically, this forces all states within the set $A$ to coalesce toward a single representative state on the Bloch sphere. The same logic applies to the set $B$. In practice, finite training data and convergence thresholds ensure that this condition is met approximately, yielding $\Theta_{AA}, \Theta_{BB} \leq \epsilon$ for a small $\epsilon$. On the other hand, in inter-set optimization, the objective function $\min_{\theta} \frac{1}{|A||B|} \sum_{i,j} |\langle a_i|b_j\rangle|^2$ is minimized when the fidelity between any state in $A$ and any state in $B$ is zero. 
The fidelity is zero if and only if the states are orthogonal, i.e., $\text{span}(A) \perp \text{span}(B)$. For quantum states of a single qubit, this orthogonality corresponds to an angular separation of $\Theta_{AB} = \pi/2$.
\end{proof}

The fundamental challenge for the verifier $\mathcal{V}$ in this interactive system is to design an efficient protocol that certifies the prover's claim of achieving the condition $\Theta_{AB} \rightarrow \pi/2$, as mentioned in Lemma~\ref{lem:orth}. This must be accomplished under two  operational constraints:
\begin{enumerate}
    \item \textbf{Black-Box Verification:} The verifier possesses zero internal knowledge of the prover's unitary $U(\vec{x},\theta)$, including its circuit structure, parameters $\theta$ or the underlying hardware.
    \item \textbf{Limited Quantum Capabilities:} The verifier's quantum operations are restricted to perform simple, destructive projective measurements on the output states provided by the prover.
\end{enumerate}

\begin{figure*}[!t]
    \centering
    \includegraphics[width=0.65\linewidth]{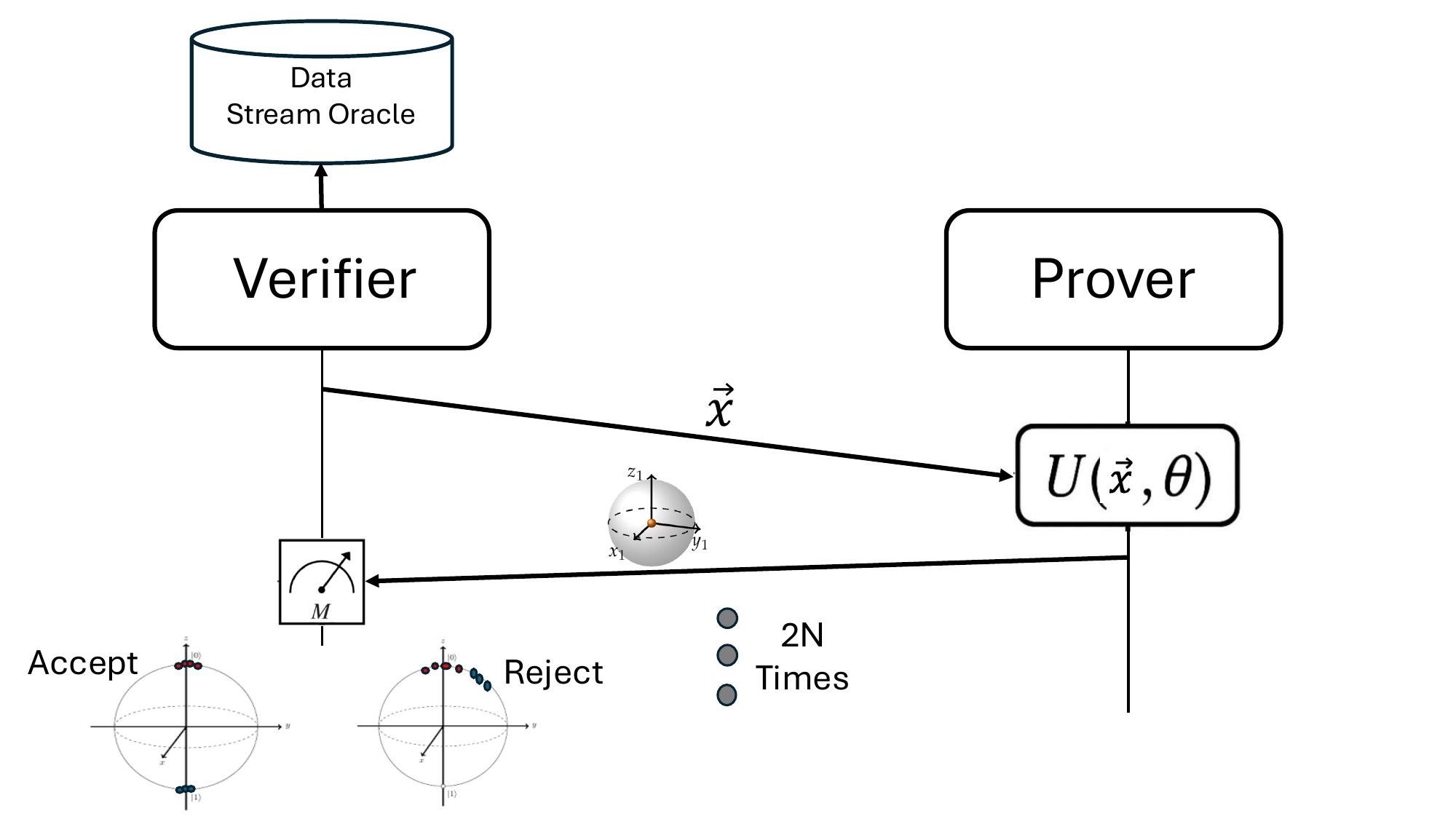}
    \caption{The proposed verification protocol.}
    \label{fig:arch}
\end{figure*}

\subsection{Protocol Details}
The Lemma~\ref{lem:orth} provides the basis for our verification procedure. A key insight is that, under optimal metric learning, all states within a single class become nearly identical. Consequently, the ensemble of qubit states for each class behaves like a single, unknown qubit state that has been prepared multiple times. This allows us to treat the two groups as two individual qubits on the Bloch sphere, and our goal reduces to estimating the angle between them. This leads to a natural verification protocol, illustrated in Figure~\ref{fig:arch}. The verifier can statistically reconstruct the quantum state of each group by collecting measurement outcomes in multiple bases and then estimate the separation angle between them. The protocol proceeds through the following sequence of steps:

\begin{enumerate}
    \item \textbf{Sampling:} The Verifier draws $2N$ samples from its data oracle: $N$ samples $\{\vec{x}_i\}$ from group $\Psi$ (label 0) and $N$ samples $\{\vec{x}_j\}$ from group $\Phi$ (label 1).

    \item \textbf{Group Allocation:} For each group ($\Psi$ and $\Phi$), the verifier divides its $N$ samples into three disjoint subgroups of size $N/3$. Each subgroup is assigned one of the three mutually unbiased measurement bases: the Standard basis $\{\ket{0}, \ket{1}\}$, the Hadamard basis $\{\ket{+}, \ket{-}\}$, or the Circular basis $\{\ket{+i}, \ket{-i}\}$.

    \item \textbf{Classical Communication:} The verifier sends the classical data points $\vec{x}$ (\textbf{without their labels}) one by one to the prover.

    \item \textbf{State Preparation:} For each received data point $\vec{x}$, the prover applies its claimed unitary embedding $U(\vec{x}, \theta)$ to the initial state $\ket{0}$ and sends the resulting qubit $\ket{\psi} = U(\vec{x}, \theta)\ket{0}$ back to the verifier.

    \item \textbf{Measurement:} The verifier measures the received qubit in the pre-determined basis corresponding to the sample's subgroup. It records the outcome (e.g., 0 for $\ket{0}$, $\ket{+}$, or $\ket{+i}$; and 1 for the opposite state).

    \item \textbf{State Reconstruction and Angle Estimation:} After all measurements have been completed, the verifier uses the recorded statistics to reconstruct the density matrix for each group and then computes the separation angle between them.
\end{enumerate}

The protocol transforms the classical measurement outputs into a reliable estimate of the inter-set separation angle, enabling a statistical claim about the embedding quality. Through this sequence of interactions, the verifier performs measurements and classical post-processing to produce an output $(flag, \hat{\Theta})$, where $flag \in \{\text{ACCEPT}, \text{REJECT}\}$ is determined by whether the estimated angle $\hat{\Theta}$ approaches $\pi/2$ within a predefined additive error $\epsilon$. 

Formally, to perform the state reconstruction, the verifier employs three mutually unbiased bases:
\begin{itemize}
    \item \textbf{Standard basis:} $\{\ket{0}, \ket{1}\}$
    \item \textbf{Hadamard basis:} $\{\ket{+}, \ket{-}\}$ where $\ket{+} = \frac{\ket{0} + \ket{1}}{\sqrt{2}}$ and $\ket{-} = \frac{\ket{0} - \ket{1}}{\sqrt{2}}$
    \item \textbf{Circular basis:} $\{\ket{+i}, \ket{-i}\}$ where $\ket{+i} = \frac{\ket{0} + i\ket{1}}{\sqrt{2}}$ and $\ket{-i} = \frac{\ket{0} - i\ket{1}}{\sqrt{2}}$
\end{itemize}

For each group, the measurement probabilities are estimated from the recorded statistics. For group $\Psi$:
\begin{align*}
    \hat{p}_{\Psi,0} &= \frac{3}{N} \sum_{i=1}^{N/3} \delta_{m_i,0} \quad \text{(Standard basis)} \\
    \hat{p}_{\Psi,+} &= \frac{3}{N} \sum_{i=N/3+1}^{2N/3} \delta_{m_i,+} \quad \text{(Hadamard basis)} \\
    \hat{p}_{\Psi,+i} &= \frac{3}{N} \sum_{i=2N/3+1}^{N} \delta_{m_i,+i} \quad \text{(Circular basis)}
\end{align*}
where $m_i$ is the measurement outcome for the $i$-th sample, and $\delta$ is the Kronecker delta. Similarly, we can obtain the probabilities for group $\Phi$:
\begin{align*}
    \hat{p}_{\Phi,0} &= \frac{3}{N} \sum_{i=1}^{N/3} \delta_{m_i,0} \quad \text{(Standard basis)} \\
    \hat{p}_{\Phi,+} &= \frac{3}{N} \sum_{i=N/3+1}^{2N/3} \delta_{m_i,+} \quad \text{(Hadamard basis)} \\
    \hat{p}_{\Phi,+i} &= \frac{3}{N} \sum_{i=2N/3+1}^{N} \delta_{m_i,+i} \quad \text{(Circular basis)}
\end{align*}

\begin{lemma}[Sufficiency]\label{lem:sufficiency}
Let $\ket{\Psi},\ket{\Phi}\in\mathbb{C}^2$ be two pure qubit states, and let 
$\rho_\Psi=\ket{\Psi}\!\bra{\Psi}$ and $\rho_\Phi=\ket{\Phi}\!\bra{\Phi}$ be their density matrices.  
Measuring each state in the three mutually unbiased bases (MUBs) of $\mathbb{C}^2$: the computational basis 
$\{\ket{0},\ket{1}\}$, the Hadamard basis $\{\ket{+},\ket{-}\}$, and the circular basis $\{\ket{i},\ket{-i}\}$, is sufficient to determine the angle $\theta$ between the states, where 
\[
\cos\theta = \big|\braket{\Psi|\Phi}\big|.
\]
\end{lemma}

\begin{proof}
Any pure qubit state $\rho = \ket{\psi}\!\bra{\psi}$ can be written in Bloch form:
\[
\rho = \frac{1}{2}\left(I + r_x \sigma_x + r_y \sigma_y + r_z \sigma_z\right),
\]
where $(r_x,r_y,r_z)$ is the Bloch vector of the state.  
Measuring $\rho$ in the three qubit MUBs yields outcome probabilities that correspond exactly to the expectation values of $\sigma_x,\sigma_y,\sigma_z$:
\[
r_x = \operatorname{tr}(\rho\,\sigma_x),\qquad
r_y = \operatorname{tr}(\rho\,\sigma_y),\qquad
r_z = \operatorname{tr}(\rho\,\sigma_z).
\]
Thus, measurements in the three MUBs provide all the information needed to fully determine the Bloch vector of each qubit state~\cite{bandyopadhyay2002new}.

Once the Bloch vectors $\vec{r}_\Psi$ and $\vec{r}_\Phi$ are determined, the fidelity between the two pure states is uniquely fixed:
\[
\big|\braket{\Psi|\Phi}\big|^2
= \operatorname{tr}(\rho_\Psi \rho_\Phi)
= \frac{1 + \vec{r}_\Psi \cdot \vec{r}_\Phi}{2}.
\]
Therefore, the angle $\theta$ between the two states, defined by 
$\cos\theta = \big|\braket{\Psi|\Phi}\big|$,  
is completely determined.  
\end{proof}

\begin{algorithm}[!t]
\small
\caption{Verifying Quantum Metric Learning.}\label{alg:angle_estimation}
\begin{algorithmic}
\Require $N$ samples from group $\Psi = \{\ket{\psi_1},...,\ket{\psi_N}\}$ and $N$ samples from group $\Phi = \{\ket{\phi_1},...,\ket{\phi_N}\}$
\Ensure Estimate $\hat{\Theta}$ of the average angle between groups

\State \textbf{Input:}
\State $\Psi = \{\ket{\psi_1},...,\ket{\psi_N}\}$ - Group $\Psi$ samples
\State $\Phi = \{\ket{\phi_1},...,\ket{\phi_N}\}$ - Group $\Phi$ samples

\State \textbf{Measurement Bases:}
\State $\mathcal{B}_0 = \{\ket{0}, \ket{1}\}$ - Standard basis
\State $\mathcal{B}_1 = \{\ket{+}, \ket{-}\}$ - Hadamard basis
\State $\mathcal{B}_2 = \{\ket{+i}, \ket{-i}\}$ - Circular basis

\State \textbf{Procedure:}
\State Randomly select $N/3$ samples from each group for each basis

\For{$i \gets 1$ to $N/3$}
    \State Measure $\ket{\psi_i}$ in $\mathcal{B}_0$, record outcome $x_i \in \{0,1\}$
    \State Measure $\ket{\phi_i}$ in $\mathcal{B}_0$, record outcome $y_i \in \{0,1\}$
\EndFor

\For{$i \gets N/3+1$ to $2N/3$}
    \State Measure $\ket{\psi_i}$ in $\mathcal{B}_1$, record outcome $x_i \in \{0,1\}$
    \State Measure $\ket{\phi_i}$ in $\mathcal{B}_1$, record outcome $y_i \in \{0,1\}$
\EndFor

\For{$i \gets 2N/3+1$ to $N$}
    \State Measure $\ket{\psi_i}$ in $\mathcal{B}_2$, record outcome $x_i \in \{0,1\}$
    \State Measure $\ket{\phi_i}$ in $\mathcal{B}_2$, record outcome $y_i \in \{0,1\}$
\EndFor

\State \textbf{Compute Group Statistics:}
\State $\hat{p}_{\Psi,0} \gets \frac{3}{N}\sum_{i=1}^{N/3} \delta_{x_i,0}$ \Comment{$\Psi$ in $\mathcal{B}_0$}
\State $\hat{p}_{\Phi,0} \gets \frac{3}{N}\sum_{i=1}^{N/3} \delta_{x_j,0}$ \Comment{$\Phi$ in $\mathcal{B}_0$}
\State $\hat{p}_{\Psi,+} \gets \frac{3}{N}\sum_{i=N/3+1}^{2N/3} \delta_{x_i,+}$ \Comment{$\Psi$ in $\mathcal{B}_1$}
\State $\hat{p}_{\Phi,+} \gets \frac{3}{N}\sum_{i=N/3+1}^{2N/3} \delta_{x_j,+}$ \Comment{$\Phi$ in $\mathcal{B}_1$}
\State $\hat{p}_{\Psi,+i} \gets \frac{3}{N}\sum_{i=2N/3+1}^{N} \delta_{x_i,+i}$ \Comment{$\Psi$ in $\mathcal{B}_2$}
\State $\hat{p}_{\Phi,+i} \gets \frac{3}{N}\sum_{i=2N/3+1}^{N} \delta_{x_j,+i}$ \Comment{$\Phi$ in $\mathcal{B}_2$}

\State \textbf{Reconstruct Density Matrices:}
\State For $\Psi$:
\State $\rho_{\Psi} \gets \begin{bmatrix}
\hat{p}_{\Psi,0} & a_{\Psi} \\
a_{\Psi}^* & 1-\hat{p}_{\Psi,0}
\end{bmatrix}$ 
\State where $a_{\Psi} = \frac{1}{2}(2\hat{p}_{\Psi,+}-1 + i(2\hat{p}_{\Psi,+i}-1))$

\State For $\Phi$:
\State $\rho_{\Phi} \gets \begin{bmatrix}
\hat{p}_{\Phi,0} & a_{\Phi} \\
a_{\Phi}^* & 1-\hat{p}_{\Phi,0}
\end{bmatrix}$ 
\State where $a_{\Phi} = \frac{1}{2}(2\hat{p}_{\Phi,+}-1 + i(2\hat{p}_{\Phi,+i}-1))$

\State \textbf{Compute Fidelity:}
\State $M \gets \sqrt{\rho_{\Psi}} \, \rho_{\Phi} \, \sqrt{\rho_{\Psi}}$
\State $\mathcal{F} \gets \left(\mathrm{Tr}(\sqrt{M})\right)^2$

\State \textbf{Estimate Angle:}
\State $\hat{\Theta} \gets \arccos\!\left(\sqrt{\mathcal{F}}\right)$

\State \textbf{Output:} $\hat{\Theta}$ (ACCEPT if $\hat{\Theta} \geq \pi/2 - \gamma$, REJECT otherwise)
\end{algorithmic}
\end{algorithm}


The verifier reconstructs the density matrix for each group using the Bloch sphere representation. For a single-qubit state, the density matrix is:

\[
\rho = \frac{1}{2}(I + \vec{r} \cdot \vec{\sigma}) = \frac{1}{2} \begin{pmatrix} 
1 + r_z & r_x - i r_y \\ 
r_x + i r_y & 1 - r_z 
\end{pmatrix},
\]

where $\vec{r} = (r_x, r_y, r_z)$ is the Bloch vector and $\vec{\sigma}$ are the Pauli matrices.
The Bloch vector components are determined from the measurement probabilities:

\begin{align*}
r_x &= 2\hat{p}_{+} - 1, & r_y &= 2\hat{p}_{+i} - 1, & r_z &= 2\hat{p}_{0} - 1
\end{align*}

Applying this to both groups yields their density matrices:
The reconstructed density matrices are:

\begin{align*}
\rho_\Psi &= \frac{1}{2} \begin{pmatrix} 
1 + r_z^\Psi & r_x^\Psi - i r_y^\Psi \\ 
r_x^\Psi + i r_y^\Psi & 1 - r_z^\Psi 
\end{pmatrix}
= \begin{pmatrix}
\hat{p}_{\Psi,0} & a_{\Psi} \\
a_{\Psi}^* & 1-\hat{p}_{\Psi,0}
\end{pmatrix} \\
\rho_\Phi &= \frac{1}{2} \begin{pmatrix} 
1 + r_z^\Phi & r_x^\Phi - i r_y^\Phi \\ 
r_x^\Phi + i r_y^\Phi & 1 - r_z^\Phi 
\end{pmatrix}
= \begin{pmatrix}
\hat{p}_{\Phi,0} & a_{\Phi} \\
a_{\Phi}^* & 1-\hat{p}_{\Phi,0}
\end{pmatrix},\\
a_{\Psi} &= \tfrac{1}{2}\left(2\hat{p}_{\Psi,+}-1 + i(2\hat{p}_{\Psi,+i}-1)\right), \\
a_{\Phi} &= \tfrac{1}{2}\left(2\hat{p}_{\Phi,+}-1 + i(2\hat{p}_{\Phi,+i}-1)\right)
\end{align*}




The separation angle between the two groups is estimated using quantum fidelity. For mixed states, the fidelity is defined as:

\[
\mathcal{F}(\rho_\Psi, \rho_\Phi) = \left( \mathrm{Tr} \sqrt{ \sqrt{\rho_\Psi} \rho_\Phi \sqrt{\rho_\Psi} } \right)^2
\]

The estimated Bures angle separation is then:

\[
\hat{\Theta}_\text{Bures} = \arccos\left( \sqrt{\mathcal{F}(\rho_\Psi, \rho_\Phi)} \right)
\]

The Bures angle $\hat{\Theta}_\text{Bures}$ provides a natural geometric interpretation of quantum state distinguishability, generalizing the concept of angular separation to mixed states. When the groups are maximally separated, $\hat{\Theta}_\text{Bures} \approx \pi/2$. The algorithm can verify \textbf{any angle} claimed by the prover, allowing us to select the threshold margin $\gamma$ adaptively based on the claimed angle as quantified in Section~\ref{sub:sim_theory}. The verifier compares this estimate to a threshold $\pi/2-\gamma$ to accept or reject the prover's claim of optimal separation.




\subsection{Completeness and Soundness Analysis}
The following two lemmas analyze the completeness and soundness of our protocol.

\begin{lemma}[Completeness]\label{lem:completeness}
For any input $x \in \mathcal{L}$, there exist a quantum witness $\ket{\psi}$ and a polynomial-time quantum verifier $V$ such that: 
\[
\Pr\!\left[\, V(x,\ket{\psi}) \text{ accepts} \,\right] \ge 1 - \epsilon(N),
\]
where $\epsilon(N)$ is a negligible function in $N$.
\end{lemma}

\begin{proof}
If the prover is honest, all states corresponding to the same class label are identical copies of the same underlying qubit state. Given $N$ independent copies, the verifier can estimate the Bloch vector of the class state with accuracy scaling as $O(1/\sqrt{N})$ by performing repeated measurements~\cite{paris2004quantum}. Since the angle $\Theta$ between the two classes depends continuously on their Bloch vectors, this estimation accuracy guarantees that the verifier's estimate $\hat{\Theta}$ differs from the true angle by at most a negligible amount $\epsilon(N)$.  
Therefore, an honest prover enables the verifier to compute an angle estimate within the acceptance threshold, leading to acceptance with probability at least $1-\epsilon(N)$.
\end{proof} 

\begin{lemma}[Soundness]\label{lem:soundness}
For any input $x \notin \mathcal{L}$ and any (possibly malicious) prover $\mathcal{P}^*$, the verifier $V$ accepts with probability at most
\[
\Pr\!\left[\, V(x,\mathcal{P}^*) \text{ accepts} \,\right] \leq \delta(N),
\]
where $\delta(N)$ is negligible in $N$.
\end{lemma}

\begin{proof}
The verifier begins by sampling $2N$ classical inputs:
\[
\{x_1,\ldots,x_N\} \sim \Psi, \qquad 
\{y_1,\ldots,y_N\} \sim \Phi ,
\]
and sends these $2N$ samples to $\mathcal{P}^*$ \emph{without revealing} which elements came from $\Psi$ and which came from $\Phi$.  Thus, from the prover's perspective, each incoming sample is drawn from the uniform mixture
\[
\frac{1}{2}\Psi + \frac{1}{2}\Phi.
\]
Because $x\notin\mathcal{L}$, no unitary embedding exists that produces two clusters with an angle close to $\pi/2$ between them. Nevertheless, the malicious prover $\mathcal{P}^*$ attempts to output states such that approximately $N$ of them resemble a ``$\Psi$ cluster'' and the remaining $N$ resemble a ``$\Phi$ cluster,'' hoping that the verifier's reconstruction of the two group states yields large angular separation. Since the prover receives the $2N$ samples in an arbitrary order and the verifier has hidden all group information, $\mathcal{P}^*$ does not know which of the $2N$ inputs came from $\Psi$ or $\Phi$. Thus, $\mathcal{P}^*$ faces a labeling problem: it attempts to assign each of the $2N$ samples to one of two clusters, while guessing which $N$ belong to the first group and which $N$ belong to the second.

Let $G$ be the hidden true labeling of the $2N$ samples, and let $\hat{G}$ be the prover's (implicit) guessed labeling used when preparing the quantum states. Any strategy for $\mathcal{P}^*$ induces a probability of correctly guessing the hidden labels. Since samples are given to $\mathcal{P}^*$ without distinguishing information, no strategy can outperform random assignment on more than a negligible fraction of inputs.  
Formally, for any (possibly entangled) strategy of $\mathcal{P}^*$, we have: 
\[
\Pr[\hat{G}=G] \leq 2^{-2N + o(N)} = \operatorname{negl}(N).
\]

Now, If $\mathcal{P}^*$ cannot correctly guess the partitions of the $2N$ samples, then its output states cannot form two internally consistent clusters. In particular, with a high probability, each  cluster produced by $\mathcal{P}^*$ contains a non-negligible mixture of true $\Psi$ and true $\Phi$ samples. Hence, when the verifier reconstructs the two cluster states using measurements in the standard, Hadamard and circular bases, the reconstructed density matrices
\[
\tilde{\rho}_\Psi, \qquad \tilde{\rho}_\Phi
\]
will each approximate a mixture of the true underlying distributions.  
Because mixtures contract quantum distances, the fidelity between them satisfies
\[
\mathcal{F}(\tilde{\rho}_\Psi,\tilde{\rho}_\Phi)
= \left( \mathrm{Tr}\sqrt{\sqrt{\tilde{\rho}_\Psi}\tilde{\rho}_\Phi\sqrt{\tilde{\rho}_\Psi}} \right)^2
\geq 1 - \operatorname{poly}^{-1}(N),
\]
and consequently their Bures angle satisfies
\[
\Theta_{\mathrm{Bures}}(\tilde{\rho}_\Psi,\tilde{\rho}_\Phi)
= \arccos\!\left(\sqrt{\mathcal{F}(\tilde{\rho}_\Psi,\tilde{\rho}_\Phi)}\right)
\leq \operatorname{poly}^{-1}(N).
\]

The verifier accepts only if the estimated separation angle is close to $\pi/2$, i.e.,
\[
\hat{\Theta}_{\mathrm{Bures}} \geq \frac{\pi}{2} - \gamma(N),
\]
for some negligible threshold $\gamma(N)$. Thus,
\begin{align}
   \Pr\!\left[\,V(x,\mathcal{P}^*) \text{ accepts}\,\right]
&\leq \Pr[ \hat{G}=G ] \\ &+ \Pr[\hat{\Theta}_{\mathrm{Bures}} \geq \tfrac{\pi}{2}-\gamma(N)] \\
&\leq \operatorname{negl}(N)
\end{align}
establishing soundness.
\end{proof}

\section{Experiment Evaluation}
\label{experiment}
This section presents a comprehensive evaluation of the proposed verification algorithm. The algorithm estimates the angle between two groups of quantum states using measurements in three mutually unbiased bases (standard, Hadamard, and circular). 

\subsection{Experiment Setup}
We start by simulating quantum metric learning models to verify the correctness of the proposed algorithm for the different separation angle between the groups. Then, we present the results of the verification of a real QAOAEmbedding quantum metric learning model with a claimed angle $0.3\pi$~\cite{lloyd2020quantum}. The experiments is run on PennyLane which is an open-source Python library specifically designed for quantum machine learning~\cite{bergholm2018pennylane}.


\subsection{Simulation Results}
\label{sub:sim_theory}
Figure~\ref{fig:angle} presents a simulation of quantum state angle estimation between two ensembles of qubit states. The data are generated by first creating random pure qubit states and then systematically generating paired states with controlled angular separations ranging from  $0$ to $\pi/2$ radians. To simulate realistic experimental conditions where qubit states within the same group may exhibit natural variations, we added a small Gaussian perturbation to the states in each individual group, representing typical quantum states in the same group. The figure shows that our estimation protocol accurately recovers the ground truth angular separation between state ensembles, as evidenced by the close alignment between estimated angles and the perfect true angles. This validation confirms the robustness of the proposed verification protocol even in the presence of realistic intra-group variations.

\begin{figure}[!t]
    \centering
    \includegraphics[width=0.7\linewidth]{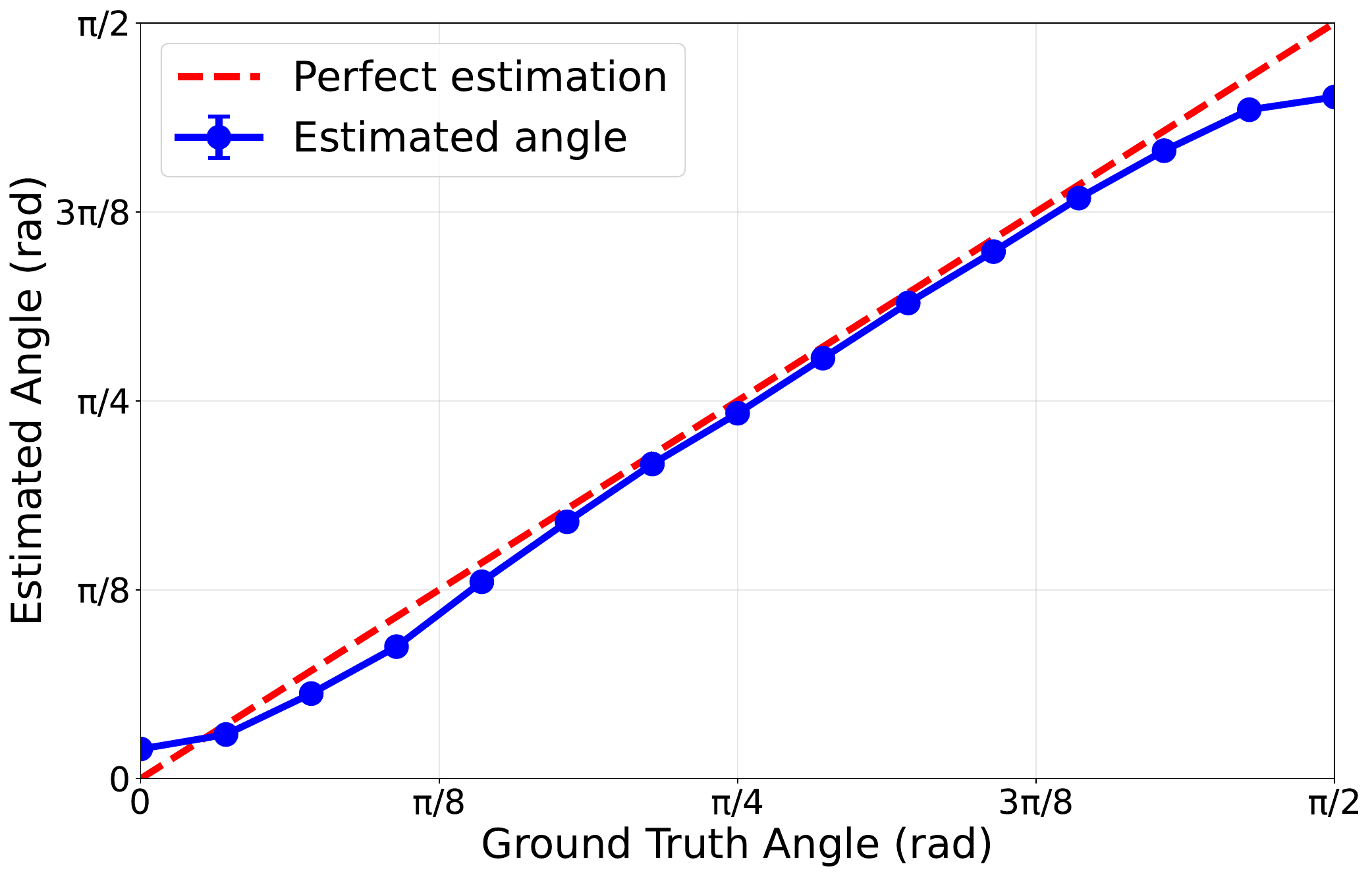}
    \caption{Angle estimation quality for the different theoretical models.}
    \label{fig:angle}
\end{figure}

\subsection{Verifying QAOAEmbedding Model}
We implement and deploy the verification protocol to evaluate a real quantum metric learning model, the PennyLane's QAOAEmbedding. This model is a variational quantum circuit that takes data from two classes and embeds them into quantum states. The model is trained using RMSPropOptimizer with step size 0.01 to maximize separation between classes in Hilbert space. Training minimizes the cost function $1 - 0.5 * (-2ab + aa + bb)$, where $aa$ and $bb$ are intra-class overlaps (within-class) and ab is the inter-class overlap (between-class), all computed via SWAP test measurements. The SWAP test is used during training to compute these overlaps, which drive optimization. Our verification algorithm reconstructs density matrices from measurements in the three bases and estimates the angle between groups. Figure~\ref{fig:samples} shows the effect of increasing the number of samples $N$ on the estimated angle. The true angle is calculated by approximating the density matrix for each group using all data samples. As expected, the figure shows that, as we increase the number of samples $N$, we obtain a better estimate of the angle between the two groups. Figure~\ref{fig:fidelity} further shows the effect of increasing the number of samples $N$ on the estimated fidelity (i.e., overlap between the two groups). The figure indicates that, as we increase the number of samples $N$, we achieve a better estimate of fidelity/overlap between the two groups.

\begin{figure}[!t]
    \centering
    \includegraphics[width=0.7\linewidth]{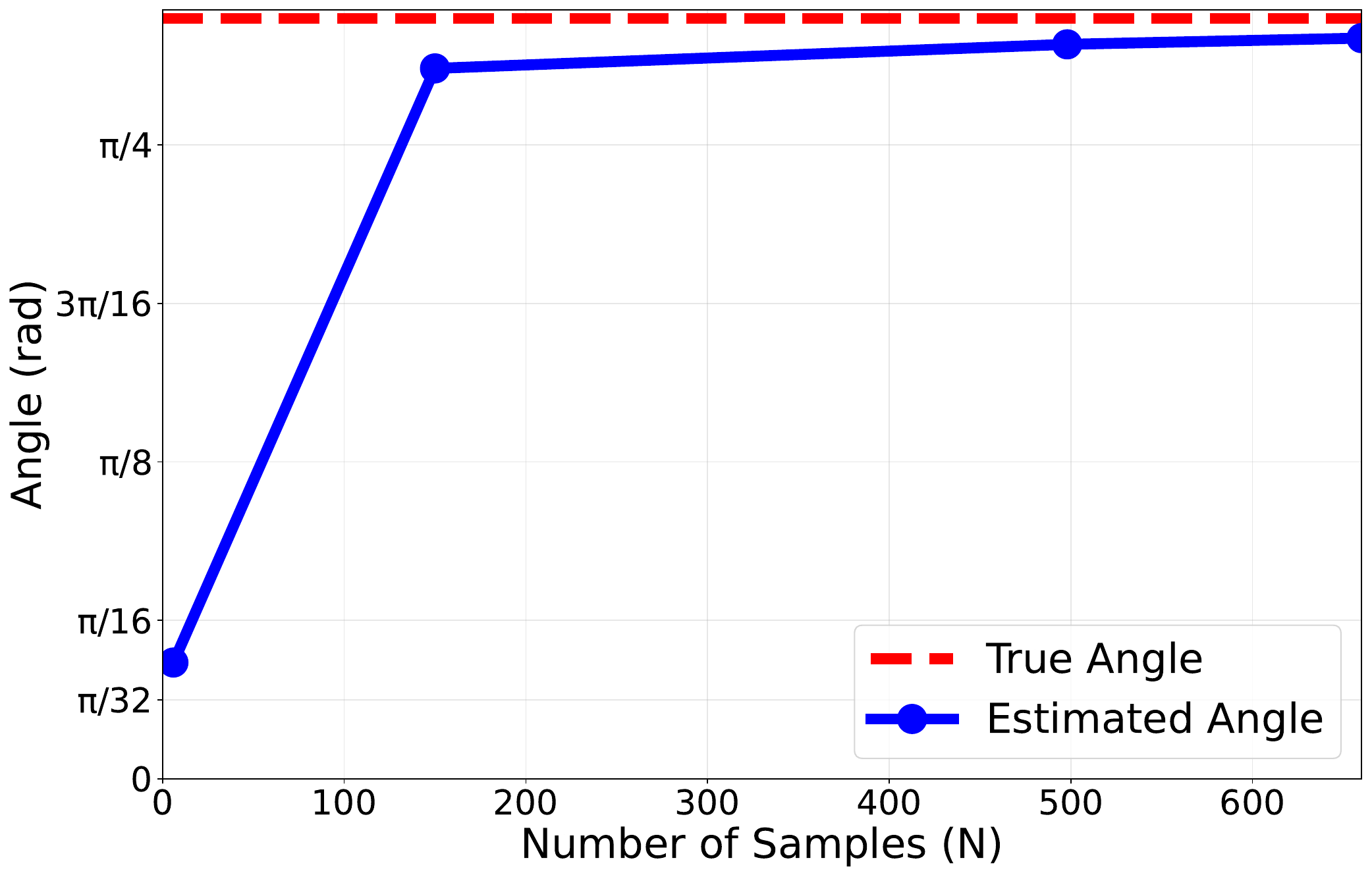}
    \caption{Effect of the number of samples $N$ on angle estimation.}
    \label{fig:samples}
\end{figure}

\begin{figure}[!t]
    \centering
    \includegraphics[width=0.7\linewidth]{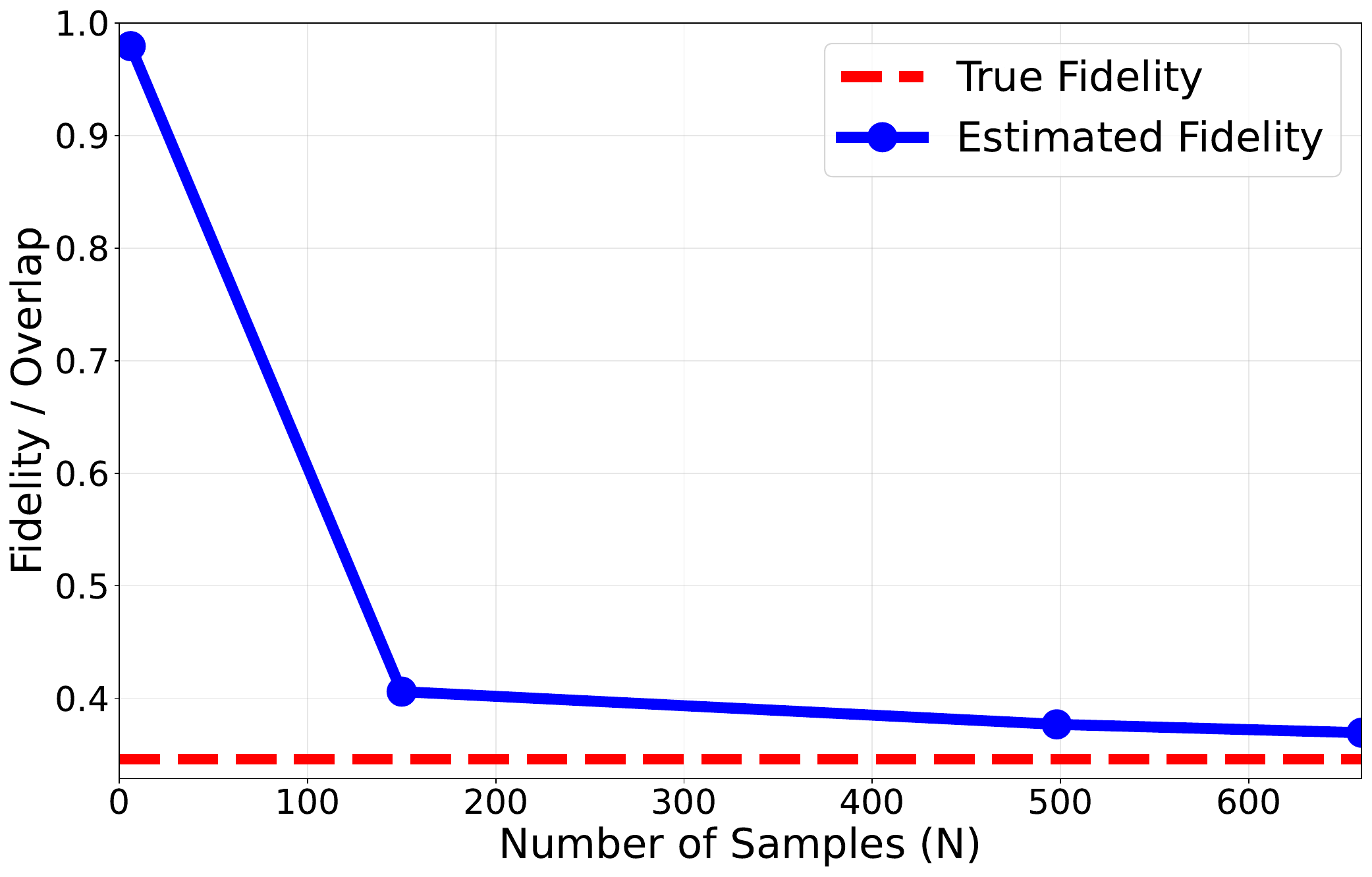}
    \caption{Effect of the number of samples $N$ on fidelity.}
    \label{fig:fidelity}
\end{figure}

\section{Discussion}
\label{discussion}
The previous results show that the proposed verification protocol can verify quantum metric learning algorithms in $O(N)$ measurements. This section discusses how to extend the proposed verification protocol to work if we have multiple groups and if we have higher-dimensional feature vectors.

\subsection{Extension to Multiple Groups}
For 2D feature vectors belonging to more than two groups, the verification procedure extends by evaluating all pairwise separations. For the $K$ groups, the verifier measures each group in the three mutually unbiased bases (computational, Hadamard, and circular) and reconstructs the corresponding density matrices $\{\rho_1, \rho_2, \ldots, \rho_K\}$. For each pair $(i,j)$, the fidelity is computed as
\[
F(\rho_i, \rho_j)
  = \left( \mathrm{Tr}\!\left[\sqrt{\sqrt{\rho_i}\,\rho_j\,\sqrt{\rho_i}} \right] \right)^2,
\]
and the angular separation is
\[
\theta_{ij} = \arccos\!\left(\sqrt{F(\rho_i,\rho_j)}\right).
\]
This yields $K(K-1)/2$ pairwise angles. Separation quality can be quantified using metrics such as the minimum pairwise angle $\min_{i,j} \theta_{ij}$, the average angle, or a threshold criterion that requires all $\theta_{ij} > \pi/2 - \gamma$. 

\subsection{Extension to Higher-Dimensional Feature Vectors}
For feature vectors with more than two components, the verification procedure extends naturally to multi-qubit quantum states.  
For an $n$-qubit system (with Hilbert space dimension $d = 2^{n}$), full quantum state reconstruction requires $d^{2}-1$~\cite{cramer2010efficient}.  
Whereas the two-dimensional case uses three mutually unbiased bases (computational, Hadamard, and circular), higher-dimensional systems require additional measurement settings. A standard approach is to perform the reconstruction using all Pauli bases. Once the states are reconstructed, the fidelity is computed following the same procedure as in the two-dimensional setting. Although the number of measurement settings required scales exponentially with $n$, for moderate system sizes, this extension enables reliable verification of quantum metric learning models in higher-dimensional feature spaces.



\section{Conclusion}
\label{conclusion}
In this paper, we present a practical blackbox verification protocol for verifying the correctness and the performance of quantum metric learning systems. We address the scenario where an untrusted prover claims to have a unitary embedding  circuit that maps classical input vectors from different groups/classes into quantum states with guaranteed orthogonal separation. The protocol enables a weak quantum verifier with a quantum measurement ability to efficiently validate whether the prover's embeddings indeed maintain the promised angular separation between the quantum representations of different sets. We implemented the proposed verification protocol and deployed it to evaluate quantum metric learning models, specifically PennyLane's QAOAEmbedding model. Our verification algorithm reconstructs density matrices from measurements and estimates separation angles using quantum fidelity. The results from both theoretical analysis as well as practical implementation reveal that the proposed protocol effectively evaluates quantum metric learning models and performs well in adversarial scenarios, demonstrating its utility as a verification tool for quantum embedding quality. This approach opens the door for verifying other quantum machine learning paradigms, including quantum clustering and classification algorithms.

\section{Acknowledgments}
This research was supported in part by a Seed Grant award from the Institute for Computational and Data Sciences at the Pennsylvania State University.

\bibliographystyle{IEEEtran}
\bibliography{references}

@article{lloyd2020quantum,
  title={Quantum embeddings for machine learning},
  author={Lloyd, Seth and Schuld, Maria and Ijaz, Aroosa and Izaac, Josh and Killoran, Nathan},
  journal={arXiv preprint arXiv:2001.03622},
  year={2020}
}

@article{cramer2010efficient,
  title={Efficient quantum state tomography},
  author={Cramer, Marcus and Plenio, Martin B and Flammia, Steven T and Somma, Rolando and Gross, David and Bartlett, Stephen D and Landon-Cardinal, Olivier and Poulin, David and Liu, Yi-Kai},
  journal={Nature communications},
  volume={1},
  number={1},
  pages={149},
  year={2010},
  publisher={Nature Publishing Group UK London}
}

@article{farhi2014quantum,
  title={A quantum approximate optimization algorithm},
  author={Farhi, Edward and Goldstone, Jeffrey and Gutmann, Sam},
  journal={arXiv preprint arXiv:1411.4028},
  year={2014}
}

@article{coles2019strong,
  title={Strong bound between trace distance and Hilbert-Schmidt distance for low-rank states},
  author={Coles, Patrick J and Cerezo, M and Cincio, Lukasz},
  journal={Physical Review A},
  volume={100},
  number={2},
  pages={022103},
  year={2019},
  publisher={APS}
}

@article{ruder2016overview,
  title={An overview of gradient descent optimization algorithms},
  author={Ruder, Sebastian},
  journal={arXiv preprint arXiv:1609.04747},
  year={2016}
}

@article{havlivcek2019supervised,
  title={Supervised learning with quantum-enhanced feature spaces},
  author={Havl{\'\i}{\v{c}}ek, Vojt{\v{e}}ch and C{\'o}rcoles, Antonio D and Temme, Kristan and Harrow, Aram W and Kandala, Abhinav and Chow, Jerry M and Gambetta, Jay M},
  journal={Nature},
  volume={567},
  number={7747},
  pages={209--212},
  year={2019},
  publisher={Nature Publishing Group UK London}
}

@article{benedetti2019parameterized,
  title={Parameterized quantum circuits as machine learning models},
  author={Benedetti, Marcello and Lloyd, Erika and Sack, Stefan and Fiorentini, Mattia},
  journal={Quantum science and technology},
  volume={4},
  number={4},
  pages={043001},
  year={2019},
  publisher={IOP Publishing}
}

@article{schuld2020circuit,
  title={Circuit-centric quantum classifiers},
  author={Schuld, Maria and Bocharov, Alex and Svore, Krysta M and Wiebe, Nathan},
  journal={Physical Review A},
  volume={101},
  number={3},
  pages={032308},
  year={2020},
  publisher={APS}
}

@article{farhi2018classification,
  title={Classification with quantum neural networks on near term processors},
  author={Farhi, Edward and Neven, Hartmut},
  journal={arXiv preprint arXiv:1802.06002},
  year={2018}
}

@article{lloyd2014quantum,
  title={Quantum principal component analysis},
  author={Lloyd, Seth and Mohseni, Masoud and Rebentrost, Patrick},
  journal={Nature physics},
  volume={10},
  number={9},
  pages={631--633},
  year={2014},
  publisher={Nature Publishing Group UK London}
}

@article{lloyd2016quantum,
  title={Quantum algorithms for topological and geometric analysis of data},
  author={Lloyd, Seth and Garnerone, Silvano and Zanardi, Paolo},
  journal={Nature communications},
  volume={7},
  number={1},
  pages={10138},
  year={2016},
  publisher={Nature Publishing Group UK London}
}

@article{rebentrost2014quantum,
  title={Quantum support vector machine for big data classification},
  author={Rebentrost, Patrick and Mohseni, Masoud and Lloyd, Seth},
  journal={Physical review letters},
  volume={113},
  number={13},
  pages={130503},
  year={2014},
  publisher={APS}
}

@book{wittek2014quantum,
  title={Quantum machine learning: what quantum computing means to data mining},
  author={Wittek, Peter},
  year={2014},
  publisher={Academic Press}
}

@article{biamonte2017quantum,
  title={Quantum machine learning},
  author={Biamonte, Jacob and Wittek, Peter and Pancotti, Nicola and Rebentrost, Patrick and Wiebe, Nathan and Lloyd, Seth},
  journal={Nature},
  volume={549},
  number={7671},
  pages={195--202},
  year={2017},
  publisher={Nature Publishing Group UK London}
}

@article{schuld2018supervised,
  title={Supervised learning with quantum computers},
  author={Schuld, Maria and Petruccione, Francesco},
  journal={Quantum science and technology},
  volume={17},
  year={2018},
  publisher={Springer}
}

@article{Havlicek2019,
  author = {Havlicek, Vojtěch and Córcoles, Antonio D. and Temme, Kristan and Harrow, Aram W. and Kandala, Abhinav and Chow, Jerry M. and Gambetta, Jay M.},
  title = {Supervised learning with quantum-enhanced feature spaces},
  journal = {Nature},
  volume = {567},
  pages = {209--212},
  year = {2019}
}

@article{Benedetti2019b,
  author = {Benedetti, Marcello and Garcia-Pintos, Delfina and Perdomo, Oscar and Leyton-Ortega, Vicente and Nam, Yunseong and Perdomo-Ortiz, Alejandro},
  title = {A generative modeling approach for benchmarking and training shallow quantum circuits},
  journal = {npj Quantum Information},
  volume = {5},
  pages = {45},
  year = {2019}
}

@book{Helstrom1976,
  author = {Helstrom, Carl W.},
  title = {Quantum detection and estimation theory},
  publisher = {Academic Press},
  year = {1976}
}

@article{Blank2019,
  author = {Blank, Carsten and Park, Daniel K. and Rhee, June-Koo Kevin and Petruccione, Francesco},
  title = {Quantum classifier with tailored quantum kernel},
  journal = {Quantum Machine Intelligence},
  volume = {1},
  pages = {1--7},
  year = {2019}
}

@article{Killoran2019,
  author = {Killoran, Nathan and Izaac, Josh and Quesada, Nicolás and Bergholm, Ville and Amy, Matthew and Weedbrook, Christian},
  title = {Strawberry Fields: A software platform for photonic quantum computing},
  journal = {Quantum},
  volume = {3},
  pages = {129},
  year = {2019}
}

@article{Harrow2017,
  author = {Harrow, Aram W. and Montanaro, Ashley},
  title = {Quantum computational supremacy},
  journal = {Nature},
  volume = {549},
  pages = {203--209},
  year = {2017}
}

@article{Schuld2019,
  author = {Schuld, Maria and Killoran, Nathan},
  title = {Quantum machine learning in feature Hilbert spaces},
  journal = {Physical Review Letters},
  volume = {122},
  number = {4},
  pages = {040504},
  year = {2019}
}

@article{Bergholm2018,
  author = {Bergholm, Ville and Izaac, Josh and Schuld, Maria and Gogolin, Christian and Killoran, Nathan},
  title = {PennyLane: Automatic differentiation of hybrid quantum-classical computations},
  journal = {arXiv preprint arXiv:1811.04968},
  year = {2018}
}

@article{helstrom1969quantum,
  title={Quantum detection and estimation theory},
  author={Helstrom, Carl W},
  journal={Journal of Statistical Physics},
  volume={1},
  pages={231--252},
  year={1969},
  publisher={Springer}
}

@article{bandyopadhyay2002new,
  title={A new proof for the existence of mutually unbiased bases},
  author={Bandyopadhyay and Boykin and Roychowdhury and Vatan},
  journal={Algorithmica},
  volume={34},
  pages={512--528},
  year={2002},
  publisher={Springer}
}

@book{paris2004quantum,
  title={Quantum state estimation},
  author={Paris, Matteo and Rehacek, Jaroslav},
  volume={649},
  year={2004},
  publisher={Springer Science \& Business Media}
}

@incollection{goldwasser1989knowledge,
  title={The knowledge complexity of interactive proof-systems},
  author={Goldwasser, Shafi and Micali, Silvio and Rackoff, Chales},
  booktitle={Providing sound foundations for cryptography: On the work of shafi goldwasser and silvio micali},
  pages={203--225},
  year={2019}
}

@article{bergholm2018pennylane,
  title={Pennylane: Automatic differentiation of hybrid quantum-classical computations},
  author={Bergholm, Ville and Izaac, Josh and Schuld, Maria and Gogolin, Christian and Ahmed, Shahnawaz and Ajith, Vishnu and Alam, M Sohaib and Alonso-Linaje, Guillermo and AkashNarayanan, B and Asadi, Ali and others},
  journal={arXiv preprint arXiv:1811.04968},
  year={2018}
}

\end{document}